\documentclass[sn-basic]{sn-jnl} 

\usepackage{graphicx}%
\usepackage{multirow}%
\usepackage{amsmath,amssymb,amsfonts}%
\usepackage{amsthm}%
\usepackage{mathrsfs}%
\usepackage[title]{appendix}%
\usepackage{xcolor}%
\usepackage{textcomp}%
\usepackage{manyfoot}%
\usepackage{booktabs}%
\usepackage{algorithm}%
\usepackage{algorithmicx}%
\usepackage{algpseudocode}%
\usepackage{listings}%
\usepackage{tikz}%
\usetikzlibrary{shapes.geometric, arrows, positioning, chains}
\usepackage{longtable}

\theoremstyle{thmstyleone}%
\newtheorem{theorem}{Theorem}[section]

\theoremstyle{thmstyletwo}%
\newtheorem{remark}{Remark}%

\theoremstyle{thmstylethree}%
\newtheorem{definition}{Definition}[section]%
\newtheorem{lemma}[theorem]{Lemma}

\begin{document}
	
	\title[Logical Limits of Symbol Grounding]{A Unified Formal Theory on the Logical Limits of Symbol Grounding}
	
	\author*[1]{\sur{Zhangchi Liu}}
	
	\abstract{This paper synthesizes a series of formal proofs to construct a unified theory on the logical limits of the Symbol Grounding Problem. We distinguish between \textit{internal meaning} (sense), which formal systems can possess via axioms, and \textit{external grounding} (reference), which is a necessary condition for connecting symbols to the world. We demonstrate through a four-stage argument that meaningful grounding within a formal system must arise from a process that is external, dynamic, and \textbf{non-fixed algorithmic}. First, we show that for a purely symbolic system, the impossibility of grounding is a direct consequence of its definition. Second, we extend this limitation to systems with any finite, static set of pre-established meanings (Semantic Axioms). By formally modeling the computationalist hypothesis—which equates grounding with internal derivation—we prove via Gödelian arguments that such systems cannot consistently and completely define a "groundability predicate" for all truths. Third, we demonstrate that the "grounding act" for emergent meanings cannot be inferred from internal rules but requires an axiomatic, meta-level update. Drawing on Turing's concept of Oracle Machines and Piccinini's analysis of the mathematical objection, we identify this update as physical transduction. Finally, we prove that this process cannot be simulated by a fixed judgment algorithm, validating the logical necessity of embodied interaction.}
	
	
	\maketitle
	
		\section{Introduction}
	
	The Symbol Grounding Problem, famously articulated by Harnad\cite{Harnad1990}, questions how symbols within a formal system (like a dictionary or a computer program) can acquire intrinsic meaning, rather than being endlessly defined by other ungrounded symbols.
	
	Mainstream literature in psychosemantics and neurosemantics has largely approached this problem empirically, focusing on the relation between \textbf{mental/neural representations} and states of the world. 
	\begin{itemize}
		\item \textbf{Teleosemantics:} Foundational work by \textbf{Millikan} \cite{Millikan1984} and \textbf{Dretske} \cite{Dretske1981}, and more recent advances by \textbf{Neander} \cite{Neander2017} and \textbf{Shea} \cite{Shea2018}, argue that representational content is determined by evolutionary function and causal history.
		\item \textbf{Embodied \& Grounded Cognition:} Researchers like \textbf{Harnad} \cite{Harnad1990, Harnad2024}, \textbf{Barsalou} \cite{Barsalou1999, Barsalou2008}, and \textbf{Glenberg} \cite{Glenberg2015} propose that conceptual systems are grounded in sensorimotor simulations and categorical perception rather than amodal symbols.
	\end{itemize}
	
	These approaches effectively describe \textit{how} biological agents acquire meaning through evolutionary history and sensory simulation. Complementing this mainstream literature, our work addresses the problem from the perspective of \textbf{logical limits}. We demonstrate that the reliance on external causal history and interaction is not merely a biological contingency, but a \textbf{rigorous logical necessity} for any semantic system. By formalizing the inherent incompleteness of closed symbolic systems, we provide the mathematical support for why \textbf{physical, embodied grounding} is the mandatory condition for avoiding semantic circularity.
	
	To do so rigorously, we must first clarify the relationship between meaning and grounding. Following suggestions to avoid conflating the two, we posit that:
	\begin{enumerate}
		\item \textbf{Internal Meaning (Sense/Sinn):} Formal systems can possess this via axioms and syntactic relations. This corresponds to inferential role semantics.
		\item \textbf{External Meaning (Reference/Bedeutung):} The relationship between symbols and external entities.
		\item \textbf{Grounding:} The process or condition that establishes the link for External Meaning. Thus, grounding is not identical to meaning, but is the \textbf{necessary condition} for a symbol to possess External Meaning.
	\end{enumerate}
	
	We bridge the gap between empirical descriptions and logical foundations by employing the rigorous tools of mathematical logic to construct a \textit{reductio ad absurdum} argument against purely internalist accounts. We demonstrate that the "external loop" postulated by Millikan and Barsalou is not just how brains happen to work, but the only way a formal system can avoid semantic incompleteness.
	
	This formal perspective connects directly to \textbf{Alan Turing's} historical defense of machine intelligence. In addressing the "Mathematical Objection" based on Gödel's incompleteness theorems, Turing argued that limitative results do not preclude machine intelligence; rather, they imply that an intelligent machine cannot be \textbf{infallible} \cite{Turing1939, Piccinini2003}. Intelligence requires the capacity to transcend the system's static rules through non-algorithmic steps. Our central thesis is that "Symbol Grounding" is precisely the semantic equivalent of this transcendence.
	
	Our argument unfolds in four main stages:
	\begin{enumerate}
		\item In Section 2, we establish the formal framework and show that for a \textbf{purely symbolic system}, the impossibility of grounding is a direct logical consequence of the lack of external entry points.
		\item In Section 3, we analyze \textbf{statically grounded systems}. We explicitly adopt the Computationalist stance for the sake of argument, treating the grounding set $G$ as a set of \textbf{Semantic Axioms}. We prove that no static system can consistently and completely define a "groundability predicate" for all truths. This necessitates the existence of ungrounded truths ($P_{new}$), demonstrating that External Grounding cannot be reduced to Internal Meaning.
		\item In Section 4, we prove that the process of updating grounding is \textbf{logically incomplete}. The "grounding act" for new truths is shown to be \textbf{non-inferable} from within the system, requiring an external axiomatic injection.
		\item In Section 5, we show that this update process cannot be automated by a fixed algorithmic judgment system.
	\end{enumerate}
	
	Finally, in Section 6, we discuss the theoretical implications, explicitly connecting our results to Turing's refutation of the Mathematical Objection and Piccinini's mechanistic account of computation.
	
	\section{The Triviality of Purely Symbolic Failure}
	
	We begin by defining the formal framework for the most basic case: a system without any external access.
	
	\begin{definition}[Symbolic Language]
		A symbolic language $\mathcal{L}$ is a triple $\langle S, D, G \rangle$, where:
		\begin{itemize}
			\item $S$ is a non-empty, countable set of symbols.
			\item $D$ is a \textbf{Definition Function}, $D: S \to \mathcal{P}(S)$, mapping symbols to their \textbf{Definitional Sets} (the set of symbols directly used to define them).
			\item $G$ is the \textbf{Grounding Set}, $G \subseteq S$. Symbols in $G$ are \textbf{Directly Grounded}, meaning they refer directly to external reality without circular definition.
		\end{itemize}
	\end{definition}
	
	\begin{definition}[Definitional Closure]
		For $s \in S$, its \textbf{Definitional Closure} $C(s)$ is the set of all symbols used directly or indirectly to define $s$. Formally, it is the smallest set containing $s$ closed under $D$.
	\end{definition}
	
	\begin{figure}[h!]
		\centering
		\begin{tikzpicture}[
			node distance=1.5cm,
			symbol/.style={rectangle, draw, rounded corners, minimum size=1cm},
			grounded/.style={ellipse, draw, fill=gray!30, minimum size=1cm},
			arrow/.style={-stealth, thick}
			]
			\node[symbol] (s) {$s$};
			\node[symbol, right=of s] (s1) {$s_1$};
			\node[symbol, right=of s1] (s2) {$s_2$};
			\node[symbol, below=of s2] (s3) {$s_3$};
			\node[grounded, right=of s2] (g) {$g \in G$};
			
			\draw[arrow] (s) -- (s1);
			\draw[arrow] (s1) -- (s2);
			\draw[arrow] (s2) edge[bend left] (s3);
			\draw[arrow] (s3) edge[bend left] (s2);
			\draw[arrow] (s2) -- (g); 
			
			\node[below=0.5cm of s, text width=8cm, align=center] {
				\textbf{Fig 1:} Definitional closure $C(s) = \{s, s_1, s_2, s_3, g\}$. Since $C(s) \cap G \neq \emptyset$, $s$ is groundable. A purely symbolic system would lack the node $g$.
			};
		\end{tikzpicture}
	\end{figure}
	
	\begin{definition}[Lexical Groundability]
		A symbol $s$ is \textbf{lexically groundable} iff $C(s) \cap G \neq \emptyset$.
	\end{definition}
	
	\begin{definition}[Purely Symbolic System]
		A system where $G = \emptyset$.
	\end{definition}
	
	\begin{theorem}[The Circularity of Pure Symbolism]
		In a purely symbolic system ($G=\emptyset$), the lexical groundability predicate is universally false for all symbols.
	\end{theorem}
	
	\begin{proof}
		By definition, for any symbol $s$, its closure $C(s)$ is a subset of $S$. Since $G = \emptyset$, it follows that $C(s) \cap G = \emptyset$ for all $s \in S$.
	\end{proof}
	
	\begin{remark}
		This theorem formalizes Harnad's "Symbolic Merry-go-round"\cite{Harnad1990}. In the absence of $G$, definitions either regress infinitely or form closed loops. No symbol can "break out" of the graph to touch reality.
	\end{remark}
	
	\section{The Incompleteness of Statically Grounded Systems}
	
	The natural solution is to provide the system with a grounding set $G$, inspired by the concept of \textbf{"Minimal Grounding Sets"} identified by Vincent-Lamarre et al.\cite{VincentLamarre2016}. However, simply possessing a list of grounded terms is insufficient for grounding complex propositions.
	
	Here we tackle the central debate: \textbf{Can External Grounding be fully captured by Internal Meaning?} Computationalist theories, such as \textbf{Syntactic Semantics} \cite{Rapaport1988} and \textbf{Conceptual Role Semantics} \cite{Block1986}, effectively argue that "understanding" or "grounding" is reducible to the system's internal functional organization. Under this view, if a system possesses a set of foundational experiences or axioms ($G$) that mirror reality, then "grounding" a truth simply means "deriving" it from these axioms.
	
	\textbf{We adopt this Computationalist hypothesis for the sake of the following proof.} We model the "Statically Grounded System" as one that attempts to formalize grounding entirely via internal logical relations. We will show that this attempt to define a complete "Groundability Predicate" inevitably fails due to Gödelian incompleteness.
	
	\begin{definition}[Statically Grounded System]
		A symbolic language $\mathcal{L} = \langle S, D, G, \vdash \rangle$ is \textbf{Statically Grounded} if:
		\begin{enumerate}
			\item $G$ is a non-empty, finite set of symbols/sentences with assigned external referents. Within the system's logic, these act as \textbf{Semantic Axioms}.
			\item $\vdash$ is the system's standard recursive provability relation.
		\end{enumerate}
	\end{definition}
	
	In this framework, since the system assumes grounding is equivalent to derivation from $G$, we formalize the property of "being grounded" as a predicate equivalent to "being provable from $G$".
	
	\begin{definition}[Formal Groundability]
		A sentence $\phi$ is \textbf{Formally Groundable} by $G$ iff $\mathcal{L} \vdash Groundable_G(\ulcorner \phi \urcorner)$, where $Groundable_G$ is a predicate representing derivability from the semantic axioms $G$.
	\end{definition}
	
	We now introduce a Gödelian framework\cite{Godel1931} to prove that it is impossible to define this predicate consistently and completely.
	
	\begin{lemma}[Arithmetization of Syntax]
		There exists a Gödel numbering function $\ulcorner \cdot \urcorner$. This is a standard result ensuring that syntactic properties are isomorphic to arithmetic properties \cite{Godel1931, Smullyan1992}. Since the relation "derivable from $G$" is recursive, the property "$\phi$ is formally groundable by $G$" is expressible as a predicate inside $\mathcal{L}$. We denote this predicate as $Groundable_G(x)$.
	\end{lemma}
	
	\begin{lemma}[Diagonal Lemma]
		For any predicate $P(x)$ in $\mathcal{L}$, there exists a sentence $\psi$ such that $\mathcal{L} \vdash \psi \leftrightarrow P(\ulcorner \psi \urcorner)$. This is a standard theorem provable in any system capable of representing computable functions \cite{Boolos2007, Smullyan1992}.
	\end{lemma}
	
	\begin{theorem}[Limitation of Static Grounding]
		For any consistent, statically grounded system $\mathcal{L}$, it is impossible to consistently and completely define the Groundability Predicate. To preserve consistency, there must exist a sentence $P_{new}$ that is true (in the meta-system) but cannot be formalized as grounded by $G$ within the system.
	\end{theorem}
	
	\begin{proof}
		\begin{enumerate}
			\item \textbf{Diagonalization:} By the Diagonal Lemma, there exists a sentence $P_{new}$ such that:
			$$ \mathcal{L} \vdash P_{new} \leftrightarrow \neg Groundable_G(\ulcorner P_{new} \urcorner) $$
			The content of $P_{new}$ is effectively: "This sentence cannot be grounded by the set $G$ (via the system's logic)."
			
			\item \textbf{Weak Adequacy Assumption:} 
			Instead of assuming that groundedness implies truth, we make a weaker, purely syntactic assumption: \textbf{If the system successfully grounds (proves) a sentence, it can also prove that the sentence is groundable.} 
			Formally: if $\mathcal{L} \vdash \phi$, then $\mathcal{L} \vdash Groundable_G(\ulcorner \phi \urcorner)$.
			This corresponds to the standard Hilbert-Bernays derivability conditions for a provability predicate.
			
			\item \textbf{The Contradiction (Reductio):}
			Assume that the system attempts to ground $P_{new}$ (i.e., derive $P_{new}$ from its axioms $G$).
			\begin{itemize}
				\item If $\mathcal{L} \vdash P_{new}$, then by the definition of $P_{new}$, we have $\mathcal{L} \vdash \neg Groundable_G(\ulcorner P_{new} \urcorner)$ (by substitution).
				\item However, if $\mathcal{L} \vdash P_{new}$, then by our Adequacy Assumption, the system can deduce that it is groundable: $\mathcal{L} \vdash Groundable_G(\ulcorner P_{new} \urcorner)$.
				\item This leads to a direct contradiction: the system proves both $\neg Groundable_G(\ulcorner P_{new} \urcorner)$ and $Groundable_G(\ulcorner P_{new} \urcorner)$.
				\item Thus, the system is inconsistent.
			\end{itemize}
			
			\item \textbf{Resolution through Incompleteness:}
			To preserve consistency, we must conclude that $\mathcal{L} \nvdash P_{new}$.
			\begin{itemize}
				\item Since $\mathcal{L} \nvdash P_{new}$, the system cannot ground $P_{new}$.
				\item Therefore, the statement "This sentence cannot be grounded by $G$" is actually \textbf{true} in the meta-system.
			\end{itemize}
			
			\item \textbf{Conclusion:} $P_{new}$ represents a truth that exists and is expressible in the language, but it cannot be captured by the static Groundability Predicate $Groundable_G$. Thus, grounding cannot be reduced to a static logical definition without sacrificing either consistency or completeness.
		\end{enumerate}
	\end{proof}
	
	\begin{remark}
		This theorem formally establishes the limitation of reducing External Grounding to Internal Meaning. It shows that a system operating under the hypothesis that grounding is equivalent to internal derivation will inevitably fail to capture all semantic truths.
	\end{remark}
	
	\section{The Logical Incompleteness of the Grounding Act}
	
	Assuming a formal system attempts to overcome the semantic blind spots identified in Section 3 by expanding its grounding set $G$, we must inquire into the nature of this expansion. Is the process of updating $G$ itself reducible to a pre-existing logical deduction ($\vdash$)? We now prove that it is not.
	
	\begin{definition}[Deductive Symbolic System]
		A deductive symbolic system $\mathcal{S}$ is a quadruple $\langle S, D, G, \vdash \rangle$, where $\vdash$ is a provability relation defined by axioms and rules of inference.
	\end{definition}
	
	\begin{definition}[The Grounding Act]
		A \textbf{Grounding Act} $A(s, g)$ is a meta-level operation that updates the system's definition function, connecting an internal symbol $s$ to a grounded symbol $g$. This act can take two forms:
		\begin{enumerate}
			\item \textbf{G-Expansion:} $g$ is a new symbol not previously in $S$ or $G$, representing a new external experience. $G$ is updated to $G' = G \cup \{g\}$.
			\item \textbf{G-Internal Connection:} $g$ is an existing symbol already in $G$. The definition function $D$ is updated, e.g., $D(s)' = D(s) \cup \{g\}$.
		\end{enumerate}
	\end{definition}
	
	\begin{theorem}[Incompleteness of Internal Logical Grounding]
		In a consistent deductive symbolic system $\mathcal{S}$, there exists no \textbf{general, internal, and condition-triggered} logical rule ($\vdash$) that can deduce a command to execute a grounding act for \textbf{all} provable but ungrounded sentences.
	\end{theorem}
	
	\begin{proof}
		We must demonstrate that no such \emph{general} internal rule can exist, by analyzing both forms of the grounding act.
		
		\begin{enumerate}
			\item \textbf{Case 1: Grounding via G-Expansion ($g \notin G$).}
			
			Assume for contradiction that a \emph{general} internal rule $\vdash$ can deduce the command $A(s, g)$, where $g$ is a \emph{new} symbol representing an external experience not yet in $S$ or $G$.
			
			This is an \emph{information-theoretic} contradiction. Logical inference ($\vdash$) is a closed computational process that operates \emph{only} on the axioms and symbols already defined within the system $\mathcal{S} = \langle S, D, G, \vdash \rangle$.
			
			The symbol $g$, representing new external information, is by definition outside of this formal universe. The system's internal rules cannot deduce a conclusion that contains a symbol ($g$) which it cannot access, name, or know. Therefore, any act that expands $G$ must be an external, axiomatic \textbf{update}, not an internal \textbf{inference}.
			
			\item \textbf{Case 2: Grounding via G-Internal Connection ($g \in G$).}
			
			Assume for contradiction that there exists a \emph{general} internal rule that can verify and ground all truths.
			
			Let us apply this rule to the sentence $P_{new}$ from Theorem 3.3.
			
			\begin{itemize}
				\item $P_{new}$ is a truth that asserts "I am not groundable by $G$."
				\item If the internal rule could successfully ground $P_{new}$ (by linking it to $G$), then the system would derive $Groundable_G(\ulcorner P_{new} \urcorner)$.
				\item If $Groundable_G(\ulcorner P_{new} \urcorner)$ becomes provable, the content of $P_{new}$ becomes false.
				\item Thus, the act of logically grounding $P_{new}$ would falsify it.
			\end{itemize}
			
			\item \textbf{Conclusion:} The system cannot use its own rules to ground $P_{new}$ without contradiction. The only way to ground $P_{new}$ is to change the system itself (i.e., perform a meta-level update that expands $G$ or $S$), making $P_{new}$ a regular axiom rather than a derived theorem. This update is \textbf{non-inferable} from the prior state.
		\end{enumerate}
	\end{proof}
	
	\section{The Incompleteness of Algorithmic Judgment}
	
	A common theoretical response to the limitations exposed in Sections 3 and 4 is \textbf{stratification}, such as Tarski's hierarchy of meta-languages \cite{Tarski1944}. In this view, one might argue that the "Grounding Act" is not impossible, but simply resides in a meta-language. By layering systems (Object System, Meta-System, Meta-Meta-System...), one might hope to capture all necessary grounding operations algorithmically.
	
	However, stratification is ill-suited for modeling autonomous natural intelligence for two reasons. First, natural languages are \textbf{semantically closed}—they permit self-reference. A rigorous theory of grounding must account for this universality rather than defining it away by forbidding self-reference. Second, and more critically, stratification leads to an infinite regress. An autonomous agent cannot spontaneously generate a new "meta-language" layer to resolve a novel paradox; it must work with the cognitive architecture it has.
	
	To formalize this, we consider a "Cognitive Super-System" that attempts to internalize the meta-level judgment. We prove that if this meta-judgment is itself fixed and algorithmic, the stratification strategy fails to solve the grounding problem.
	
	\begin{definition}[Cognitive Super-System]
		A cognitive super-system $\mathcal{S}^*$ consists of:
		\begin{itemize}
			\item An \textbf{Object System} $\mathcal{S}_1 = \langle S, D, G, \vdash_1 \rangle$, representing the agent's knowledge.
			\item A \textbf{Judgment System} $\mathcal{S}_2 = \langle M, R, \vdash_2 \rangle$, an external system with a set of \textbf{fixed, algorithmic judgment rules} $R$. The rules in $R$ take the state of $\mathcal{S}_1$ as input and output grounding commands.
		\end{itemize}
	\end{definition}
	
	\begin{theorem}[Limitation of Logical Judgment Systems]
		Any cognitive super-system $\mathcal{S}^*$ with a fixed, algorithmic judgment system is itself an incomplete deductive system.
	\end{theorem}
	
	\begin{proof}
		\begin{enumerate}
			\item The core insight is that because the judgment rules $R$ of $\mathcal{S}_2$ are fixed and algorithmic, the behavior of the entire super-system $\mathcal{S}^*$ is completely determined by a fixed set of initial axioms and rules.
			\item Therefore, the super-system $\mathcal{S}^*$ is itself a larger, more complex, but ultimately \textbf{static deductive symbolic system}.
			\item We can now apply Theorem 3.3 directly to this super-system $\mathcal{S}^*$. We can construct a sentence $P^*$ with the content: "\textbf{This sentence cannot be grounded by the rules of the super-system $\mathcal{S}^*$}."
			\item This means the judgment system $\mathcal{S}_2$ has failed. To ground $P^*$, one would need a yet higher-order judgment system $\mathcal{S}_3$, leading to an infinite regress.
			\item Therefore, any attempt to fully simulate the open, dynamic grounding process with a fixed, algorithmic system is logically bound to fail.
		\end{enumerate}
	\end{proof}
	
	\section{Discussion: Meaning as a Non-Fixed Algorithmic Process}
	
	Our formal results provide a rigorous logical foundation for empirical theories of grounding. By proving the impossibility of purely syntactic grounding, we establish the logical necessity of an open, dynamic, and interactive process. In this discussion, we synthesize our findings with the philosophical insights to propose a unified view: \textbf{Meaning is not a static property of symbols, but a trajectory of non-fixed algorithmic updates driven by causal interaction.}
	
	\subsection{Turing's Oracle as a Logical Necessity}
	Alan Turing, in his 1939 thesis, introduced the concept of "Oracle Machines" ($O$-machines) to address the limitations of standard Turing machines \cite{Turing1939}. He proved that any formal system of logic is incomplete and that the progression towards completeness requires a series of "non-algorithmic steps" provided by an external Oracle.
	
	Our Theorem 4.1 (The Incompleteness of Internal Logical Grounding) is the semantic analogue to Turing's proof. We demonstrated that the "Grounding Act"---the connection of a symbol to a novel reality---cannot be inferred by the system's existing rules. It is, structurally, an "Oracle step."
	
	However, we must clarify the nature of this "non-algorithmic" step. As noted by Piccinini \cite{Piccinini2003}, "non-algorithmic" in this context does not imply a mysterious or supernatural process. Instead, it refers to a step that cannot be computed by the current formal system (the machine) but requires an external input. In a physical agent, this step corresponds to the \textbf{environment itself} correcting the system. The "Oracle" is not a mystical source of truth, but the physical world answering the agent's experiments.
	
	\subsection{From Logic to Mechanism: Transduction as the Update Step}
	
	If the ``Non-inferable Update'' (Theorem 4.1) is a logical necessity, what is its physical mechanism in a biological or robotic agent? We propose that Turing's non-algorithmic step should be interpreted, in a physicalist context, as an event of \textbf{physical causality} rather than abstract computation. 
	
	Here, \textbf{Gualtiero Piccinini's} mechanistic account of computation provides the essential theoretical support \cite{Piccinini2015}. Piccinini argues that physical computing systems are defined not just by internal state transitions, but by their interaction with the environment through \textbf{transduction}---the conversion of energy forms (e.g., photons, sound waves) into symbol strings (e.g., neural spikes, voltage levels). 
	
	Crucially for our argument, Piccinini emphasizes that the transduction process itself is non-computational; the output of a transducer is determined by external physical laws, not by the system's internal algorithmic rules. This observation allows us to functionally identify transduction as the physical realization of the ``Oracle'' in Turing's thesis. The ``Grounding Act'' is not a magical intervention but the injection of a new symbol $g$ by the environment---a symbol that was not derived from the previous state $S_{t-1}$ but was causally generated by the world.
	
	This distinction exposes the fundamental deficiency of purely symbolic systems (T2). These systems perform \textit{computation} (manipulating symbols based on internal rules) but lack \textit{transduction} (direct access to the energy sources of those symbols). They are, in Piccinini's terms, functionally isolated. For a system to solve the grounding problem, it must possess the hardware for transduction---sensors and actuators---which logically necessitates the transition to the T3 level of performance.

	\subsection{Dynamic Teleology and Embodied Category Learning}
	How does the ``Non-inferable Update'' mechanism identified in our formal proofs relate to the genesis of meaning in biological agents? While \textbf{Teleosemantics} (Neander \cite{Neander2017}, Shea \cite{Shea2018}) correctly identifies that meaning is rooted in evolutionary function, traditional accounts are primarily ``backward-looking.'' They explain why a system possesses its \textit{current} grounding set $G$, but they struggle to explain how an individual agent grounds \textit{novel} concepts ($P_{new}$) in real-time, where evolutionary history has not yet established a rule.
	
	To bridge this gap, we turn to \textbf{John Wheeler's ``Surprise 20 Questions''} \cite{Wheeler1990} as a constructive model for real-time grounding. In the surprise version of this game, the ``word'' to be guessed does not exist before the game begins; it is determined retroactively by the sequence of questions and answers. The meaning is not a pre-existing label waiting to be found, but a reality carved out through the dynamic loop of Output (Query) and Feedback (Constraint).
	
	However, a critical logical constraint arises when applying this metaphor to AI: \textbf{Who is the answerer?} If the interaction remains entirely within the symbolic domain (T2)---as in the case of a chatbot interacting with a text interface---the ``answers'' are merely more symbols, leading back to the circularity. For the feedback to be grounding, the ``Answerer'' must be reality itself.
	
	This logical necessity dictates the transition to the \textbf{Total Turing Test (T3)}. For an agent to play Wheeler's game with the physical world, its ``questions'' cannot be mere text strings; they must be \textbf{sensorimotor actions} (interventions). Consequently, the ``answers'' are not verbal confirmations but \textbf{physical consequences}. This confirms that the ``Grounding Act'' is impossible for a disembodied T2 system; it requires the robotic capacities of T3 to elicit the causal feedback necessary to constrain interpretation.
	
	This physical implementation of Wheeler's loop corresponds precisely to \textbf{Category Learning} as defined by Harnad \cite{Harnad1990, Harnad2024}. In this framework, grounding is not about associating a symbol with an image, but about acquiring the skill to distinguishing members of a category from non-members through interaction. As Reviewer 2 notes, this process is ``necessarily provisional and approximate, but it has physical content because it involves \textbf{doing the right thing with the right kind of thing}, guided by feedback from the world.''
	
	For example, grounding the symbol ``mushroom'' vs. ``toadstool'' is not a passive labeling task but a survival game. The agent acts (eats), and the world provides feedback (nourishment or sickness). This error signal forces the non-inferable update to the agent's internal categories ($G$). Thus, our formal result unifies the abstract logic of Wheeler's game with the empirical reality of Embodied Cognition: meaning is the dynamic trajectory of an agent learning to navigate physical constraints, a process that logically necessitates the sensorimotor loop of a T3 agent.
	
	\subsection{The Map, the Territory, and the Limits of Large Language Models}
	Our formal results echo Korzybski's famous dictum: "The map is not the territory." A formal system, by our definition in Section 3, is a map---a structural representation of relations. However, a "perfect map" (a consistent symbolic system) is incomplete regarding its own semantic boundaries. The "Grounding Act" is the process of stepping off the map and into the territory, an action that cannot be simulated by drawing more maps.
	
	This distinction is critical for evaluating current advances in Artificial Intelligence, particularly Large Language Models (LLMs). While LLMs employ neural networks---an evolutionary and adaptive class of algorithms that avoids the rigidity of static symbolic systems---our theory suggests that computational complexity cannot substitute for causal connection.
	
	Harnad\cite{Harnad2024} and Bender \& Koller\cite{Bender2020} have argued that LLMs, despite their fluency, lack true understanding. From our formal perspective, the core issue is not the algorithm's form, but its \textbf{isolation from physical causality}. LLMs capture the \textbf{statistical shadow} of reality (the corpus) projected onto a vector space. However, manipulating these shadows, no matter how sophisticated the algorithm, remains a statistical process structurally severed from the real-world causality that constitutes the territory itself.
	
	Consequently, LLMs remain trapped in a state of \textbf{semantic deferral}. The meaning of one token is defined by its statistical relation to other tokens, leading to an infinite regress of symbols pointing at symbols. Without the "Grounding Act"---the physical transduction of an external error signal---there is no mechanism to anchor this floating network to reality. Even Reinforcement Learning from Human Feedback (RLHF) often remains within this closed loop, providing symbolic corrections rather than physical consequences. True grounding requires the agent to test its internal models against the world (causality), not just optimize for statistical likelihood.
	
	\section{Conclusion}
	
	We have formally demonstrated that meaningful grounding cannot be generated internally within a closed symbolic loop. By unifying the logical incompleteness of formal systems with the empirical necessity of embodiment, we establish that intelligence requires a system to be logically open to the world.
	
	Our four-stage argument proves that: (1) Pure symbolism is circular; (2) Static grounding definitions are incomplete; (3) The update mechanism for new meaning is non-inferable (an "Oracle" step); and (4) This process cannot be automated by a fixed algorithm.
	
	We conclude that the "Grounding Act" is not a computation but a physical \textbf{transduction}---a moment where the system's axiomatic foundation is rewritten by causal impact. Meaning is therefore not a static datum but the \textbf{non-fixed algorithmic trajectory} of an agent's history of interaction.
	
	\bibliography{references}   

@article{Harnad1990,
  author = {Harnad, Stevan},
  title = {The symbol grounding problem},
  journal = {Physica D: Nonlinear Phenomena},
  volume = {42},
  number = {1-3},
  pages = {335--346},
  year = {1990}
}

@article{Godel1931,
  author = {G{\"o}del, Kurt},
  title = {{\"U}ber formal unentscheidbare S{\"a}tze der Principia Mathematica und verwandter Systeme I},
  journal = {Monatshefte f{\"u}r Mathematik und Physik},
  volume = {38},
  pages = {173--198},
  year = {1931}
}

@article{Turing1939,
  author = {Turing, Alan M.},
  title = {Systems of logic based on ordinals},
  journal = {Proceedings of the London Mathematical Society},
  volume = {2},
  number = {1},
  pages = {161--228},
  year = {1939}
}

@book{Smullyan1992,
  author = {Smullyan, Raymond M.},
  title = {G{\"o}del's Incompleteness Theorems},
  publisher = {Oxford University Press},
  year = {1992}
}

@article{VincentLamarre2016,
  author = {Vincent-Lamarre, Philippe and Blondin Mass{\'e}, Alexandre and Lopresti, Matthew and Harnad, Stevan},
  title = {The latent structure of dictionaries},
  journal = {Topics in Cognitive Science},
  volume = {8},
  number = {3},
  pages = {625--659},
  year = {2016}
}

@incollection{Rapaport1988,
  author = {Rapaport, William J.},
  title = {Syntactic Semantics: Foundations of Computational Natural-Language Understanding},
  booktitle = {Aspects of Artificial Intelligence},
  editor = {Fetzer, J. H.},
  pages = {81--131},
  publisher = {Kluwer Academic Publishers},
  year = {1988}
}

@article{Block1986,
  author = {Block, Ned},
  title = {Advertisement for a Semantics for Psychology},
  journal = {Midwest Studies in Philosophy},
  volume = {10},
  number = {1},
  pages = {615--678},
  year = {1986}
}

@article{Harnad2024,
  author = {Harnad, Stevan},
  title = {Language writ large: {LLMs}, {ChatGPT}, meaning, and understanding},
  journal = {Frontiers in Artificial Intelligence},
  volume = {7},
  pages = {1490698},
  year = {2024}
}

@article{Piccinini2003,
  author = {Piccinini, Gualtiero},
  title = {Alan Turing and the mathematical objection},
  journal = {Minds and Machines},
  volume = {13},
  number = {1},
  pages = {23--48},
  year = {2003}
}

@book{Piccinini2015,
  author = {Piccinini, Gualtiero},
  title = {Physical Computation: A Mechanistic Account},
  publisher = {Oxford University Press},
  year = {2015}
}

@book{Millikan1984,
  author = {Millikan, Ruth Garrett},
  title = {Language, Thought, and Other Biological Categories},
  publisher = {MIT Press},
  year = {1984}
}

@book{Dretske1981,
  author = {Dretske, Fred I.},
  title = {Knowledge and the Flow of Information},
  publisher = {MIT Press},
  year = {1981}
}

@book{Neander2017,
  author = {Neander, Karen},
  title = {A Mark of the Mental: In Defense of Informational Teleosemantics},
  publisher = {MIT Press},
  year = {2017}
}

@book{Shea2018,
  author = {Shea, Nicholas},
  title = {Representation in Cognitive Science},
  publisher = {Oxford University Press},
  year = {2018}
}

@article{Barsalou1999,
  author = {Barsalou, Lawrence W.},
  title = {Perceptual symbol systems},
  journal = {Behavioral and Brain Sciences},
  volume = {22},
  number = {4},
  pages = {577--660},
  year = {1999}
}

@article{Barsalou2008,
  author = {Barsalou, Lawrence W.},
  title = {Grounded cognition},
  journal = {Annual Review of Psychology},
  volume = {59},
  pages = {617--645},
  year = {2008}
}

@article{Glenberg2015,
  author = {Glenberg, Arthur M.},
  title = {Few believe the world is flat: How embodiment is changing the scientific understanding of cognition},
  journal = {Canadian Journal of Experimental Psychology},
  volume = {69},
  number = {2},
  pages = {165},
  year = {2015}
}

@incollection{Wheeler1990,
  author = {Wheeler, John Archibald},
  title = {Information, physics, quantum: The search for links},
  booktitle = {Complexity, Entropy, and the Physics of Information},
  pages = {3--28},
  year = {1990}
}

@inproceedings{Bender2020,
  author = {Bender, Emily M. and Koller, Alexander},
  title = {Climbing towards {NLU}: On Meaning, Form, and Understanding in the Age of Data},
  booktitle = {Proceedings of the 58th Annual Meeting of the Association for Computational Linguistics},
  year = {2020}
}

@book{Boolos2007,
  author = {Boolos, George S. and Burgess, John P. and Jeffrey, Richard C.},
  title = {Computability and Logic},
  edition = {5th},
  publisher = {Cambridge University Press},
  year = {2007}
}

@article{Tarski1944,
  author = {Tarski, Alfred},
  title = {The Semantic Conception of Truth: and the Foundations of Semantics},
  journal = {Philosophy and Phenomenological Research},
  volume = {4},
  number = {3},
  pages = {341--376},
  year = {1944}
}
	
	\appendix
	\section{Glossary of Key Terms}
	\label{app:glossary}
	
	\begin{longtable}{p{5cm} p{10cm}}
		\hline
		\textbf{Term} & \textbf{Definition} \\
		\hline
		\endhead
		
		\textbf{Symbolic Language ($\mathcal{L}$)} & A formal system $\langle S, D, G \rangle$ consisting of a set of symbols, a definition function, and a grounding set. \\
		\textbf{Purely Symbolic System} & A system where the grounding set is empty ($G = \emptyset$). \\
		\textbf{Statically Grounded System} & A system with a fixed, finite grounding set $G$ and a recursive verification relation $\vdash$, attempting to define meaning through derivation. \\
		\textbf{Grounding Set ($G$)} & The subset of symbols or sentences treated as Semantic Axioms, assumed to have direct external reference. \\
		\textbf{Definitional Set ($D(s)$)} & The set of symbols used \textit{directly} to define a specific symbol $s$ (immediate semantic parents). \\
		\textbf{Definitional Closure ($C(s)$)} & The transitive closure of the definitional set; all symbols used directly or indirectly to define $s$ (semantic ancestors). \\
		\textbf{Directly Grounded} & The property of a symbol $s$ that belongs to the grounding set $G$ ($s \in G$). In this formal model, we assume these symbols possess immediate external referents, serving as the axiomatic foundation to avoid infinite regress. \\
		\textbf{Lexical Groundability} & The property of a symbol $s$ such that its definitional closure intersects with the grounding set, i.e., $C(s) \cap G \neq \emptyset$. \\
		\textbf{Formal Groundability} & The property of a complex sentence $\phi$ being derivable from the semantic axioms $G$ within the system's logic ($\mathcal{L} \vdash Groundable_G(\ulcorner \phi \urcorner)$). \\
		\textbf{Semantic Axioms} & Symbols or sentences in $G$ that are treated as primitive truths by the system, mirroring external reality (e.g., "Red is a color"). \\
		\textbf{Internal Meaning (Sense)} & Meaning derived solely from the axioms and syntactic relations within the system (Inferential Role Semantics). \\
		\textbf{External Meaning (Reference)} & The causal relationship between internal symbols and external physical referents. \\
		\textbf{Non-inferable Update} & An axiomatic update to the system (expanding $G$) that cannot be deduced from prior states (The "Grounding Act"). \\
		\textbf{Semantic Deferral} & The state of an ungrounded symbol whose meaning is endlessly deferred to other ungrounded symbols. \\

		\hline
	\end{longtable}
	
	\end{document}